\documentclass[12pt]{article}

\usepackage{amsmath,amssymb,amsthm}
\usepackage{authblk}
\usepackage{graphicx}
\usepackage{xcolor}
\usepackage{enumitem}
\usepackage{multicol}
\usepackage{tikz}
\usepackage[colorlinks,citecolor=blue,linkcolor=blue]{hyperref}

\newtheorem{theorem}{Theorem}[section]
\newtheorem{lemma}[theorem]{Lemma}
\newtheorem{corollary}[theorem]{Corollary}

\newtheorem{definition}[theorem]{Definition}
\newtheorem{remark}[theorem]{Remark}
\newtheorem{condition}[theorem]{Condition}

\title{
High-Dimensional Signal Compression:
\\
Lattice Point Bounds and Metric Entropy}

\author[1]{Alex Iosevich \thanks{A. Iosevich was supported in part by NSF DMS-2154232}}
\author[2$\mathsection$]{Armen Vagharshakyan\thanks{A. Vagharshakyan was supported by the Higher Education and Science Committee of RA (Research Project No 24RL-1A028)}}
\author[3]{Emmett Wyman\thanks{E. Wyman was supported in part by NSF DMS-2422900}}

\affil[1]{\small{University of Rochester, Rochester, NY, USA}}
\affil[2]{
\small{Institute of Mathematics and
Yerevan State University, Yerevan, Armenia}
}
\affil[3]{\small{
Binghamton University, Binghamton, NY, USA
}}
\affil[$\mathsection$]{\small{Corresponding author: \textit {avaghars@kent.edu}}}

\begin{document}
\maketitle

\begin{abstract}
We study worst-case signal compression under an $\ell^2$ energy constraint,
with coordinate-dependent quantization precisions.
The compression problem is reduced to counting lattice points in a diagonal ellipsoid.
Under balanced precision profiles, we obtain explicit, dimension-dependent upper bounds on the logarithmic codebook size.
The analysis refines Landau's classical lattice point estimates using uniform Bessel bounds due to Olenko and explicit Abel summation.


\end{abstract}

\section{Introduction}

High-dimensional compression and quantization problems arise throughout signal processing, information theory, and data science, particularly in settings where signals are subject to global energy constraints but individual coordinates may be represented with different accuracies.
Such variable-precision models occur naturally when different signal components carry unequal importance, are measured at different resolutions, or are constrained by heterogeneous hardware or sensing costs.

In this paper we study a deterministic, worst-case signal compression problem.
Given a signal $f:\{1,\dots,k\}\to\mathbb{R}$ satisfying an $\ell^2$ energy constraint
\[
\sum_{i=1}^k f(i)^2 \le R^2,
\]
and given coordinate-dependent quantization precisions $\varepsilon_1,\dots,\varepsilon_k>0$, we ask for an explicit upper bound on the number of distinct quantized representations that can arise in the worst case.
Equivalently, we seek a sharp bound on the logarithmic codebook size required to represent all signals obeying the energy constraint at the prescribed coordinate-wise accuracy.

Our analysis is fully nonasymptotic in the sense of probability: no distributional assumptions are made on the signal class, and no average-case arguments are used.
Instead, the problem is reduced to a geometric and number-theoretic question---counting integer lattice points in a high-dimensional ellipsoid determined by the quantization profile.
This reduction is elementary, but it allows tools from analytic number theory to be brought to bear on a compression problem.

From an information-theoretic perspective, the quantity $\ln \tau(k)$ represents a worst-case covering number, or metric entropy, for the class of admissible signals under coordinate-wise quantization.
It can be viewed as a deterministic analogue of rate--distortion complexity, specialized to worst-case guarantees.
Our objective is to make the dependence of this quantity on the dimension $k$ explicit, since this regime is central in modern high-dimensional applications.

The main result of the paper shows that, under a mild balanced-precision condition,
\[
\frac{1}{k}\ln \tau(k)
\le
\ln R - \ln\!\left(\frac{\varepsilon_{\mathrm{total}}}{k}\right)
- \frac12 \ln k + O(1),
\]
where $\varepsilon_{\mathrm{total}}=\sum_{i=1}^k \varepsilon_i$.
Alternatively, if $V$ denotes the volume of the associated ellipsoid, then 
\[
\tau(k)\leq e^{Ck}V,\]
where $C$ is an absolute constant.
The factor $e^{Ck}$ is tempered in regimes where the growth of $k$ is constrained by $R.$ 

We emphasize that the bound above is obtained in a regime where the dimension $k$ and the energy parameter $R$ are coupled, and under a mild balanced-precision assumption on the quantization profile. These hypotheses are natural in the worst-case compression setting considered here: without a comparability condition on the $\varepsilon_i$, the lattice count can grow arbitrarily large even when the total precision $\varepsilon_{\mathrm{total}}$ is fixed, and without a restriction relating $k$ and $R$, no dimension-explicit control of the entropy is possible.

From this perspective, the main contribution of this paper is a dimension-explicit entropy bound for ellipsoidal quantization under structural assumptions that prevent such degeneracies. The proof combines classical lattice-point methods with uniform analytic estimates in a way that makes the dependence on $k$ completely explicit.

Technically, the proof develops a dimension-explicit refinement of a classical lattice point estimate of Landau, in which all error terms are tracked as functions of the ambient dimension. This is achieved by combining Landau's difference-operator framework with uniform bounds on Bessel functions due to Olenko and a quantified Abel summation argument. The resulting estimate yields a sharp upper bound on the logarithmic codebook size, with explicit dependence on $k$ that is not visible in the classical asymptotic theory.

The present work is part of a broader effort to understand how geometric constraints govern compressibility and complexity in high-dimensional signal models.
In related work, effective dimension is captured through analytic invariants such as the Fourier ratio of coefficient vectors, leading to deterministic bounds on recovery, localization, and rate--distortion complexity.
Although the technical tools differ, both approaches reflect the same guiding principle: compressibility is controlled not by ambient dimension alone, but by geometric structure---manifesting here through lattice-point geometry induced by energy and quantization constraints.

\subsection*{Organization}

Section II formulates the compression problem and reduces it to counting lattice points in a diagonal ellipsoid.
Section III recalls Landau's difference formula and introduces the necessary notation.
Sections IV--VI develop the dimension-explicit estimates using Abel summation and uniform Bessel bounds, culminating in the proof of the main theorem.

\section{Compression}
\label{section:compression}

\subsection{The class $\mathcal{X}$}
\label{subsection:classX}

Fix $k\in\mathbb{N}$ and $R\geq 1$. Let $\mathcal{X}$ be the class of 
real-valued signals
\[
f:\{1,2,\dots,k\}\to\mathbb{R}
\]
satisfying the energy constraint
\begin{equation}\label{variation}
\sum_{i=1}^k f(i)^2\leq R^2.
\end{equation}
If $g$ is a discrete-time signal rooted at $0$, meaning $g(0)=0$, then one may take $f(i)=g(i+1)-g(i)$ as a forward difference. 

Working with forward differences rather than the values themselves is a classical idea in approximation theory and numerical analysis.

The bounded energy condition \eqref{variation} can be interpreted as bounded mean-square variation of the underlying signal.

\subsection{Problem setting}

Suppose we are given a family of functions $\mathcal{C}\subset\mathcal{X}$. We assume no a priori information on $\mathcal{C}$ beyond $\mathcal{C}\subset\mathcal{X}$. We also fix a sequence of precisions $\varepsilon_1,\dots,\varepsilon_k>0$. The dependence on $i$ allows certain coordinates to be recorded more accurately than others.

\subsection{The compression procedure}

For each $f\in\mathcal{C}$ and each $i$ we record $f(i)$ with precision $\varepsilon_i$ by rounding down to an integer multiple of $\varepsilon_i$. Define
\begin{align}\label{floor}
u_i=\left\lfloor \frac{f(i)}{\varepsilon_i}\right\rfloor,\text{ if }f(i)\geq 0,
\\\nonumber
u_i=\left\lceil \frac{f(i)}{\varepsilon_i}\right\rceil,\text{ if }f(i)< 0.
\end{align}
This produces an integer vector $(u_1,\dots,u_k)\in\mathbb{Z}^k$. As $f$ ranges over $\mathcal{C}$ these vectors form the compressed representation of $\mathcal{C}$.

\subsection{The recovery procedure}
Given a code vector $(u_1,\dots,u_k)$, the coordinate values of $f$ satisfy

\begin{align*}
u_i\varepsilon_i\leq f(i)<(u_i+1)\varepsilon_i, \text{ if }u_i\geq 0,
\\
(u_i-1)\varepsilon_i\leq f(i)<u_i\varepsilon_i, \text{ if }u_i<0.
\end{align*}
Thus, in all cases,
the vector $(f(1),\dots,f(k))$ lies in an axis-parallel parallelepiped whose side lengths are $\varepsilon_i$.

\subsection{The compression problem}

In the worst-case setting where the only information is $\mathcal{C}\subset\mathcal{X}$, we ask for an upper bound on the number of distinct code vectors that can arise. This is the size of a codebook sufficient to represent any family $\mathcal{C}\subset\mathcal{X}$.

\subsection{Reduction to a number theory problem}

From \eqref{floor} we have 
$\left|u_i\varepsilon_i\right|\leq |f(i)|$,
hence
\[
(u_i\varepsilon_i)^2\leq f(i)^2.
\]
Summing and using \eqref{variation} yields
\begin{equation}\label{integers}
\sum_{i=1}^k (u_i\varepsilon_i)^2\leq R^2.
\end{equation}
Let $M$ be the diagonal matrix with entries
\begin{equation}\label{m}
M_{ii}=\varepsilon_i^2.
\end{equation}
Then \eqref{integers} is equivalent to $\vec u^T M\vec u\leq R^2$.

\begin{definition}[Ellipsoidal codebook size]
Define
\begin{equation}\label{tau}
\tau(k)=\sharp\left\{\vec u\in\mathbb{Z}^k: \vec u^T M\vec u\leq R^2\right\}.
\end{equation}
\end{definition}

The quantity $\ln\tau(k)$ is the natural worst-case logarithmic codebook size. It is also a quantitative proxy for compactness, since finiteness of $\tau(k)$ for fixed $(k,R,\varepsilon_i)$ corresponds to the fact that the admissible codes form a finite set.

Bounds of this type are relevant whenever one seeks dimension-explicit guarantees for worst-case quantization or lossy compression without probabilistic assumptions.

We also define the total precision
\begin{equation}\label{total}
\varepsilon_{\mathrm{total}}=\sum_{i=1}^k \varepsilon_i.
\end{equation}

We also define the geometric mean precision
\begin{equation}\label{geom}
\varepsilon_{\mathrm{geom}}=\left(\prod_{i=1}^k \varepsilon_i\right)^{1/k}.
\end{equation}

The dependence of $\tau(k)$ on $(\varepsilon_i)$ is controlled by the geometric mean $\varepsilon_{\mathrm{geom}}$ rather than by $\varepsilon_{\mathrm{total}}$ alone. In order to express bounds in terms of $\varepsilon_{\mathrm{total}}$ one needs an additional comparability hypothesis, recorded below.

\begin{condition}[Balanced precision profile]\label{cond:balanced}
There exists a constant $C\geq 1$ such that for all $i$,
\begin{equation}\label{balanced}
\frac{1}{C}\frac{\varepsilon_{\mathrm{total}}}{k}\leq \varepsilon_i\leq C\frac{\varepsilon_{\mathrm{total}}}{k}.
\end{equation}
\end{condition}

\begin{remark}
Condition \ref{cond:balanced} replaces the informal statement that a worst-case bound for certain upper estimates is achieved when all $\varepsilon_i$ are equal. Without a condition of this type, fixing only $\varepsilon_{\mathrm{total}}$ permits some $\varepsilon_i$ to be arbitrarily small, and the lattice count \eqref{tau} can become arbitrarily large.
\end{remark}
We prove the following upper estimate in a regime 
 where
the growth of $k$ is constrained by 
$R.$ Specifically,
\begin{equation}\label{regime}
    k\leq c_0 R^{1+\frac{1}{k}}
\end{equation}
for an absolute constant $c_0.$
\begin{theorem}\label{thm:main}
Assume $k\geq 2$ and $R\geq 2$. Then
\begin{equation}\label{mainbound}
\frac{\ln \tau(k)}{k}\leq \ln R-\ln\varepsilon_{\mathrm{geom}}-\frac{1}{2}\ln k+O(1),
\end{equation}
where $O(1)$ is bounded as $k\to\infty$. If, in addition, Condition \ref{cond:balanced} holds, then $\varepsilon_{\mathrm{geom}}$ is comparable to $\varepsilon_{\mathrm{total}}/k$, and \eqref{mainbound} can be rewritten as
\begin{equation}\label{mainbound-balanced}
\frac{\ln \tau(k)}{k}\leq \ln R-\ln\left(\frac{\varepsilon_{\mathrm{total}}}{k}\right)-\frac{1}{2}\ln k+O(1).
\end{equation}
The implicit constants may depend on the constant $C$ in Condition \ref{cond:balanced}.
\end{theorem}

\begin{corollary}[Covering number interpretation]\label{cor:covering}
Assume the hypotheses of Theorem \ref{thm:main}. Then the number $\tau(k)$ is a worst-case covering number for $\mathcal X$ under coordinate-wise quantization at scales $\varepsilon_1,\dots,\varepsilon_k$. In particular, the set $\mathcal X$ can be covered by at most $\tau(k)$ axis-parallel boxes of side lengths $\varepsilon_i$.
\end{corollary}

\subsection{Relation to counting lattice points and metric entropy}\label{subsection:metric-entropy}

For fixed $k,$ Landau \cite{EL} provides an asymptotic estimate for the number of lattice points in a $k$-dimensional ellipsoid as its volume grows.
 In the rest of the article, we revisit Landau's argument to track its dependence on $k.$ For this we combine Landau's framework with Olenko's \cite{O} uniform estimate on Bessel functions.
 
 We refer to
 \cite{N} for an alternate proof of Landau's estimate and to \cite{IKKN} for a general survey on counting lattice points in large regions. 

An important application of these counting problems is refining the asymptotics in Weyl's law set in a prefixed dimension (see e.g. \cite{I},\cite{IW}). In contrast, high-dimensional compression - the application that we consider - does require obtaining dimension-dependent bounds.

The quantity $\ln\tau(k)$ can be viewed as an entropy bound in the sense of Kolmogorov and Tikhomirov \cite{KT}, specialized to the present coordinate-wise quantization model. The ellipsoid constraint in \eqref{integers} identifies the set of feasible integer codes, and counting its lattice points gives a sharp worst-case upper bound on the number of distinct quantized signals compatible with the energy constraint \eqref{variation}.

\section{Formula (38) of \cite{EL}}
\label{formula-34}

Rather than reproducing Landau's argument verbatim, we extract and adapt the components needed for our setting, with the goal of making all dependence on the dimension explicit. In particular, we reorganize the argument to isolate the contributions of the main term and error terms in a way that is compatible with the compression problem under consideration.

We inspect Landau \cite{EL} in order to estimate the right-hand side of \eqref{tau} from above. We refer to formula (38) on page 473 of \cite{EL}, which states
\begin{multline}\label{38}
\Delta[B_{\rho}(x)]
=\frac{\pi^{\frac{k}{2}}}{\sqrt{D}\Gamma\left(\frac{k}{2}+\rho+1\right)}
\,\gamma\,\Delta\left[x^{\frac{k}{2}+\rho}\right]
+\eta\,\Delta\left[x^{\rho}\right] \\
+\frac{c}{\pi^{\rho-\frac{k}{2}}}
\sum_{n=1}^{\infty}
\frac{\alpha_n}{\lambda_n^{\frac{k}{4}+\frac{\rho}{2}}}
\Delta\left[
x^{\frac{k}{4}+\frac{\rho}{2}}
J_{\frac{k}{2}+\rho}\left(2\pi\sqrt{\lambda_n x}\right)
\right].
\end{multline}

We record the notations used in \eqref{38}.

\begin{itemize}
\item Following page 470 of \cite{EL},
\begin{equation}\label{470}
B_{\rho}(x)=\frac{1}{\rho !}\sum_{l_n\leq x}(x-l_n)^{\rho}.
\end{equation}
Here $B_\rho(x)$ is a smoothed version of the lattice point count, with smoothing order $\rho$.

\item Following page 465 of \cite{EL}, $0<l_1\leq l_2\leq \cdots$ is a non-decreasing rearrangement of the positive values that the positive-definite quadratic form
\[
Q(u_1,\dots,u_k)=\sum_{\mu,\nu=1}^k a_{\mu\nu}u_{\mu}u_{\nu}
\]
takes on $\mathbb{Z}^k$, counted with multiplicity. The function
\begin{equation}\label{b}
B(x)=B_0(x)=\sum_{l_n\leq x}1
\end{equation}
counts the lattice points in the ellipsoid $\{\vec u: Q(\vec u)\leq x\}$.

\item Let $M$ be the symmetric positive-definite matrix corresponding to $Q$. Then
\begin{equation}\label{mgeneral}
M=\begin{pmatrix}
a_{11} & a_{12} & \dots & a_{1k}\\
a_{21} & a_{22} & \dots & a_{2k}\\
\vdots & \vdots & \ddots & \vdots\\
a_{k1} & a_{k2} & \dots & a_{kk}
\end{pmatrix}.
\end{equation}
Following line 17 on page 459 of \cite{EL}, $D$ denotes
\begin{equation}\label{d}
D=\det(M).
\end{equation}

\item Following page 468 of \cite{EL},
\begin{equation}\label{rho}
\rho=\left\lfloor \frac{k}{2}+1\right\rfloor.
\end{equation}

\item Following page 465 line 19 of \cite{EL}, $\gamma=1$.

\item Following point 2 on page 459 of \cite{EL}, $k\geq 2$.

\item Following formula (12) on page 463 of \cite{EL}, $J_{\nu}$ denotes the Bessel function
\begin{equation}\label{bessel}
J_{\nu}(x)=\sum_{\lambda=0}^{\infty}\frac{(-1)^{\lambda}}{\lambda !\Gamma(\nu+\lambda+1)}\left(\frac{x}{2}\right)^{\nu+2\lambda}.
\end{equation}

\item Following line 17 on page 465 of \cite{EL}, in the special case relevant here we take $\alpha_n=1$.

\item Following line 18 on page 465 of \cite{EL}, $0<\lambda_1\leq \lambda_2\leq \cdots$ is a non-decreasing rearrangement of the positive values, with multiplicity, taken by the dual quadratic form $\bar{Q}$ on $\mathbb{Z}^k$. Landau defines
\begin{equation}\label{459}
\bar{Q}(u_1,\dots,u_k)=\sum_{\mu,\nu=1}^k \frac{1}{D}\frac{\partial D}{\partial a_{\mu\nu}}u_{\mu}u_{\nu}.
\end{equation}
Let $\bar{M}$ be the symmetric matrix corresponding to $\bar{Q}$. (Note by Cramer's rule, $\bar M = M^{-1}$.) If $M$ is diagonal, then
\begin{equation}\label{easymatrix}
(\bar{M})_{ii}=\frac{1}{a_{ii}}.
\end{equation}
In our diagonal case \eqref{m} we obtain
\begin{equation}\label{finalmente-m}
(\bar{M})_{ii}=\varepsilon_i^{-2}.
\end{equation}
Since $B(x)$ counts $\vec u\in\mathbb{Z}^k$ with $Q(\vec u)\leq x$, we have
\begin{equation}\label{relation}
\tau(k)=B(R^2).
\end{equation}

\item Following line 7 on page 471 of \cite{EL},
\[
\eta=\frac{\zeta(0)}{\rho !}.
\]
Here $\zeta$ is the Epstein zeta function corresponding to $M$. We use the standard fact that $\zeta(0)=-1.$ 

We refer to \cite{B} for a survey of elementary properties of $\zeta,$ including its definition via quadratic forms, the functional equation and the relationship to the dual quadratic form,  to \cite{CS} for its relationship to the geometry of ellipsoids and lattices, and to \cite{BBG} for computing its values.

\item Following the last line on page 466 of \cite{EL},
\begin{equation}\label{cdef}
c=D^{-\frac{1}{2}}\pi^{-\frac{k}{2}}.
\end{equation}

\item Following the second to last line of page 473 of \cite{EL}, we take $x\geq 1$ and define
\begin{equation}\label{z}
z=x^{\frac{1}{k+1}}.
\end{equation}

\item Following line 9 on page 472 of \cite{EL}, the operator $\Delta$ is the symmetric difference operator of order $\rho$,
\begin{equation}\label{delta}
(\Delta[F])(x)=\sum_{\nu=0}^{\rho}(-1)^{\rho-\nu}C_{\rho}^{\nu}F(x+\nu z).
\end{equation}
The operator $\Delta$ acts as a discrete averaging operator over an interval of length $\rho z$.
\end{itemize}

\section{A rough entropy estimate}
\label{page-465}

Following page 465 of \cite{EL}, define
\[
G(x)=\sum_{\lambda_n\leq x}1.
\]
Then $G(x)$ counts lattice points in the dual ellipsoid
\begin{equation}\label{ellipsoid}
\bar{\mathcal{E}}=\left\{\vec u\in\mathbb{R}^k: \vec u^T\bar{M}\vec u\leq x\right\}.
\end{equation}
Landau records a bound of the form $G(x)=O(x^{k/2})$. For later constant-tracking we quantify this bound in the diagonal case. We will start with a very rough estimate for $G(x)$ by bounding the ellipsoid in a box. Though crude, it is a necessary starting point for the argument.

Using \eqref{finalmente-m}, the condition $\vec u^T\bar{M}\vec u\leq x$ becomes
\[
\sum_{i=1}^k \varepsilon_i^{-2} u_i^2\leq x.
\]
Therefore $\bar{\mathcal{E}}$ is contained in the axis-parallel box
\[
\mathcal{B}=\prod_{i=1}^k\left[-\sqrt{x} \varepsilon_i,\sqrt{x} \varepsilon_i\right].
\]
The number of integer points in $\mathcal{B}$ equals
\begin{equation}\label{two}
\sharp(\mathcal{B}\cap\mathbb{Z}^k)=\prod_{i=1}^k\left(2\left\lfloor\sqrt{x}\varepsilon_i\right\rfloor+1\right).
\end{equation}
Consequently
\begin{equation}\label{three}
G(x)\leq \prod_{i=1}^k\left(2\sqrt{x}\varepsilon_i+1\right).
\end{equation}

Condition \ref{cond:balanced} implies that each $\varepsilon_i$ is comparable to $\varepsilon_{\mathrm{total}}/k$. In particular, there is a constant $c_0\geq 1$ depending only on $C$ such that
\[
\varepsilon_i\leq c_0\frac{\varepsilon_{\mathrm{total}}}{k}.
\]
If $x\geq 1$, we have $1 \leq \sqrt{x}$. A simple unified bound is
\[
2\sqrt{x}\varepsilon_i+1 \leq 3\sqrt{x}\,\max(1,\varepsilon_i).
\]
Using $\varepsilon_i\leq c_0\varepsilon_{\mathrm{total}}/k$ gives
\[
2\sqrt{x}\varepsilon_i+1 \leq 3\sqrt{x}\,\max\!\left(1, c_0\frac{\varepsilon_{\mathrm{total}}}{k}\right).
\]
Thus \eqref{three} yields
\[
G(x)\leq \left[3\sqrt{x}\,\max\!\left(1, c_0\frac{\varepsilon_{\mathrm{total}}}{k}\right)\right]^k.
\]
We then have the bound
\begin{equation}\label{g}
G(x)\leq 2^k\sigma^k x^{\frac{k}{2}},
\end{equation}
where we can take
\begin{equation}\label{sigma}
\sigma=c_1\max\!\left(1,\frac{\varepsilon_{\mathrm{total}}}{k}\right) + \frac{c_2}{\sqrt{x}}.
\end{equation}
The constants $c_1,c_2$ depend only on $C$ from Condition \ref{cond:balanced}. In the regime $\varepsilon_{\mathrm{total}}/k\ll 1$ and $x\gg 1$, the quantity $\sigma$ is small.

\section{Abel summation}
\label{page-474}

Landau uses Abel summation to estimate sums over $\lambda_n$ in terms of $G$. The qualitative statements in \cite{EL} can be recorded as
\begin{equation}\label{original-comment-1}
\sum_{\lambda_n\leq y}\frac{1}{\lambda_n^w}=O\left(y^{\frac{k}{2}-w}\right).
\end{equation}
This holds for $w<k/2$. The complementary estimate is
\begin{equation}\label{original-comment-2}
\sum_{\lambda_n> y}\frac{1}{\lambda_n^w}=O\left(\frac{1}{y^{w-\frac{k}{2}}}\right).
\end{equation}
This holds for $w>k/2$.

For later explicit bounds we use the quantified Abel summation inequalities that follow from the argument on page 474 of \cite{EL}. For $y\geq \lambda_1$ and $w<k/2$,
\begin{equation}\label{new-comment}
\sum_{\lambda_n\leq y}\frac{1}{\lambda_n^w}\leq w\int_{\lambda_1}^{y}\frac{G(u)}{u^{w+1}}du+\frac{G(\lambda_{G(y)})}{\lambda_{G(y)}^w}.
\end{equation}
For $w>k/2$,
\begin{equation}\label{nnew-comment}
\sum_{\lambda_n> y}\frac{1}{\lambda_n^w}\leq w\int_{y}^{\infty}\frac{G(u)}{u^{w+1}}du.
\end{equation}

\section{Olenko's uniform estimate}
\label{olenko-section}

We use a uniform estimate for Bessel functions due to Olenko \cite{O}.

\begin{theorem}[Olenko]\label{thm:olenko}
For $\nu>0$,
\begin{equation}\label{obessel}
\sup_{x\geq 0}\sqrt{x}\left|J_{\nu}(x)\right|\leq b\sqrt{\nu^{1/3}+\frac{\alpha_1}{\nu^{1/3}}+\frac{3\alpha_1^2}{10\nu}}.
\end{equation}
Here $b=0.674885\dots$ and $\alpha_1=1.855757\dots$.
\end{theorem}

We now record a convenient monotone envelope.

\begin{lemma}\label{lem:olenko-envelope}
Define $f(\nu)$ for $\nu\geq 3$ as follows. For $3\leq \nu\leq 5$, let
\[
f(\nu)=\sup_{3\leq t\leq 5}b\sqrt{t^{1/3}+\frac{\alpha_1}{t^{1/3}}+\frac{3\alpha_1^2}{10t}}.
\]
For $\nu\geq 5$ let
\[
f(\nu)=b\sqrt{\nu^{1/3}+\frac{\alpha_1}{\nu^{1/3}}+\frac{3\alpha_1^2}{10\nu}}.
\]
Then $f$ is non-decreasing on $[3,\infty)$ and
\begin{equation}\label{mybessel}
\sup_{x\geq 0}\sqrt{x}\left|J_{\nu}(x)\right|\leq f(\nu).
\end{equation}
Moreover there exist absolute constants $c_1,c_2>0$ such that with $\rho$ defined by \eqref{rho} one has
\begin{equation}\label{fbound}
c_1\leq \left[\frac{f(\rho)}{f\left(\frac{k}{2}+\rho\right)}\right]^{\frac{2}{\rho}}\leq c_2.
\end{equation}
\end{lemma}

\section{The $\Delta$ operator}
\label{operator-delta}

We quantify the estimates used by Landau around formulas (35), (36), and (37) on    472--473 of \cite{EL}. We keep the notations $k,\rho,z$ from \eqref{rho} and \eqref{z}. We work with the term
\[
x^{\frac{k}{4}+\frac{\rho}{2}}J_{\frac{k}{2}+\rho}\left(2\pi\sqrt{\lambda_n x}\right).
\]

From \eqref{delta} and the triangle inequality,
\begin{multline}\label{35a}
\left|\Delta\left(x^{\frac{k}{4}+\frac{\rho}{2}}J_{\frac{k}{2}+\rho}\left(2\pi\sqrt{\lambda_n x}\right)\right)\right| \\
\leq \sum_{\nu=0}^{\rho}C_{\rho}^{\nu}\left|(x+\nu z)^{\frac{k}{4}+\frac{\rho}{2}}J_{\frac{k}{2}+\rho}\left(2\pi\sqrt{\lambda_n(x+\nu z)}\right)\right|.
\end{multline}

Using \eqref{mybessel} with order $\frac{k}{2}+\rho$ and argument $2\pi\sqrt{\lambda_n(x+\nu z)}$ yields
\[
\left|J_{\frac{k}{2}+\rho}\left(2\pi\sqrt{\lambda_n(x+\nu z)}\right)\right|
\leq \frac{1}{\sqrt{2\pi}}\frac{f\left(\frac{k}{2}+\rho\right)}{(\lambda_n(x+\nu z))^{1/4}}.
\]
Substituting this into \eqref{35a} and using $\sum_{\nu=0}^{\rho}C_{\rho}^{\nu}=2^{\rho}$ gives
\begin{multline}\label{35b}
\left|\Delta\left(x^{\frac{k}{4}+\frac{\rho}{2}}J_{\frac{k}{2}+\rho}\left(2\pi\sqrt{\lambda_n x}\right)\right)\right| 
\leq \frac{2^{\rho}}{\sqrt{2\pi}}\frac{f\left(\frac{k}{2}+\rho\right)}{\lambda_n^{1/4}}\max_{0\leq \nu\leq \rho}(x+\nu z)^{\frac{k}{4}+\frac{\rho}{2}-\frac{1}{4}}.
\end{multline}
Since $x+\nu z\leq x+\rho z$,
\begin{multline}\label{35c}
\left|\Delta\left(x^{\frac{k}{4}+\frac{\rho}{2}}J_{\frac{k}{2}+\rho}\left(2\pi\sqrt{\lambda_n x}\right)\right)\right| 
\leq \frac{2^{\rho}}{\sqrt{2\pi}}\frac{f\left(\frac{k}{2}+\rho\right)}{\lambda_n^{1/4}}(x+\rho z)^{\frac{k}{4}+\frac{\rho}{2}-\frac{1}{4}}.
\end{multline}

Landau also records a mean value identity on page 473 of \cite{EL}. There exists $\xi$ with $x\leq \xi\leq x+\rho z$ such that
\[
\Delta\left(x^{\frac{k}{4}+\frac{\rho}{2}}J_{\frac{k}{2}+\rho}\left(2\pi\sqrt{\lambda_n x}\right)\right)
=\pi^{\rho}\lambda_n^{\rho/2}z^{\rho}\xi^{k/4}J_{k/2}\left(2\pi\sqrt{\lambda_n\xi}\right).
\]
Applying \eqref{mybessel} with order $k/2$ gives
\[
\left|J_{k/2}\left(2\pi\sqrt{\lambda_n\xi}\right)\right|
\leq \frac{1}{\sqrt{2\pi}}\frac{f(k/2)}{(\lambda_n\xi)^{1/4}}.
\]
Since $f$ is non-decreasing and $\rho\geq 1$ we have $f(k/2)\leq f(k/2+\rho)$. Using $\xi\leq x+\rho z$ we obtain
\begin{multline}\label{36}
\left|\Delta\left(x^{\frac{k}{4}+\frac{\rho}{2}}J_{\frac{k}{2}+\rho}\left(2\pi\sqrt{\lambda_n x}\right)\right)\right| 
\leq \frac{\pi^{\rho}}{\sqrt{2\pi}}z^{\rho}\lambda_n^{\frac{\rho}{2}-\frac{1}{4}}(x+\rho z)^{\frac{k}{4}-\frac{1}{4}}f\left(\frac{k}{2}+\rho\right).
\end{multline}

We combine \eqref{35c} and \eqref{36} by taking the minimum of the two upper bounds. This yields a quantified analogue of Landau's formula (37) on page 473 of \cite{EL}:
\begin{multline}\label{37}
\left|\Delta\left(x^{\frac{k}{4}+\frac{\rho}{2}}J_{\frac{k}{2}+\rho}\left(2\pi\sqrt{\lambda_n x}\right)\right)\right| \\
\leq \frac{1}{\sqrt{2\pi}}\frac{(x+\rho z)^{\frac{k}{4}-\frac{1}{4}}}{\lambda_n^{1/4}}
\min\Biggl[
2^{\rho}(x+\rho z)^{\frac{\rho}{2}}f\left(\frac{k}{2}+\rho\right), 
\pi^{\rho}\lambda_n^{\frac{\rho}{2}}z^{\rho}f\left(\frac{k}{2}+\rho\right)
\Biggr].
\end{multline}

We now estimate the infinite series term in \eqref{38}. Using \eqref{37},
\[
\left|\sum_{n=1}^{\infty}\frac{1}{\lambda_n^{\frac{k}{4}+\frac{\rho}{2}}}\Delta\left[x^{\frac{k}{4}+\frac{\rho}{2}}J_{\frac{k}{2}+\rho}\left(2\pi\sqrt{\lambda_n x}\right)\right]\right|
\]
is bounded by
\begin{multline*}
\sum_{n=1}^{\infty}\frac{1}{\lambda_n^{\frac{k}{4}+\frac{\rho}{2}}}\frac{1}{\sqrt{2\pi}}\frac{(x+\rho z)^{\frac{k}{4}-\frac{1}{4}}}{\lambda_n^{1/4}}
\\
\cdot \min\left[2^{\rho}(x+\rho z)^{\frac{\rho}{2}}f\left(\frac{k}{2}+\rho\right),\pi^{\rho}\lambda_n^{\frac{\rho}{2}}z^{\rho}f\left(\frac{k}{2}+\rho\right)\right].
\end{multline*}
Factor out the common prefactor and split the sum at a threshold $y$. Define $y$ by the equality of the two terms inside the minimum. This gives
\[
2^{\rho}(x+\rho z)^{\rho/2}=\pi^{\rho}y^{\rho/2}z^{\rho}.
\]
Equivalently,
\begin{equation}\label{ydef}
y=\frac{4(x+\rho z)}{\pi^2 z^2}.
\end{equation}
With this choice we obtain the bound
\begin{multline}\label{splitbound}
\left|\sum_{n=1}^{\infty}\frac{1}{\lambda_n^{\frac{k}{4}+\frac{\rho}{2}}}\Delta\left[x^{\frac{k}{4}+\frac{\rho}{2}}J_{\frac{k}{2}+\rho}\left(2\pi\sqrt{\lambda_n x}\right)\right]\right| \\
\leq \frac{(x+\rho z)^{\frac{k}{4}-\frac{1}{4}}}{\sqrt{2\pi}}f\left(\frac{k}{2}+\rho\right)
\Biggl[
\pi^{\rho}z^{\rho}\sum_{\lambda_n\leq y}\frac{1}{\lambda_n^{\frac{k}{4}+\frac{1}{4}}} 
+2^{\rho}(x+\rho z)^{\frac{\rho}{2}}\sum_{\lambda_n> y}\frac{1}{\lambda_n^{\frac{k}{4}+\frac{\rho}{2}+\frac{1}{4}}}
\Biggr].
\end{multline}

\begin{lemma}\label{lem:yest}
There exist absolute constants $c_3,c_4>0$ such that for $x\geq 1$,
\begin{equation}\label{yestimate}
c_3 x^{1-\frac{2}{k+1}}\leq y\leq c_4 x^{1-\frac{2}{k+1}}.
\end{equation}
Moreover, Landau notes (comment 1 on page 474 of \cite{EL}) that
\begin{equation}\label{stau}
\lambda_{G(y)}\leq y.
\end{equation}
The inequality \eqref{stau} holds because $G(y)$ counts eigenvalues $\leq y$, so the $G(y)$-th eigenvalue is at most $y$.
\end{lemma}

\begin{proof}
By \eqref{z} and \eqref{ydef},
\[
y=\frac{4(x+\rho x^{1/(k+1)})}{\pi^2 x^{2/(k+1)}}.
\]
Since $\rho\leq k+1$ and $x\geq 1$, the ratio $(x+\rho x^{1/(k+1)})/x$ is bounded above and below by positive constants. This yields \eqref{yestimate}. The estimate \eqref{stau} is Landau's observation cited above.
\end{proof}

We now plug \eqref{new-comment} and \eqref{nnew-comment} into the two sums in \eqref{splitbound}. We also use the rough estimate \eqref{g}. For compactness, define three integrals which appear from Abel summation.

\begin{lemma}\label{lem:integrals}
Assume Condition \ref{cond:balanced}. Let $I_1,I_2,I_3$ be defined by
\begin{equation}\label{I1def}
I_1=\int_{\lambda_1}^{y}\frac{G(u)}{u^{\frac{k}{4}+\frac{5}{4}}}\,du.
\end{equation}
\begin{equation}\label{I2def}
I_2=\frac{G(\lambda_{G(y)})}{\lambda_{G(y)}^{\frac{k}{4}+\frac{1}{4}}}.
\end{equation}
\begin{equation}\label{I3def}
I_3=\int_{y}^{\infty}\frac{G(u)}{u^{\frac{k}{4}+\frac{\rho}{2}+\frac{5}{4}}}\,du.
\end{equation}
Then there exist constants $c_6,c_7,c_8$ such that for $x\geq 1$,
\begin{align}
I_1 &\leq c_6 2^k\sigma^k y^{\frac{k}{4}-\frac{1}{4}}, \label{I1bd}\\
I_2 &\leq c_7 2^k\sigma^k y^{\frac{k}{4}-\frac{1}{4}}, \label{I2bd}\\
I_3 &\leq c_8 2^k\sigma^k y^{\frac{k}{4}-\frac{\rho}{2}-\frac{1}{4}}. \label{I3bd}
\end{align}
\end{lemma}

\begin{proof}
We use \eqref{g}. Since $G(u)\leq 2^k\sigma^k u^{k/2}$, the integral defining $I_1$ satisfies
\[
I_1\leq 2^k\sigma^k\int_{\lambda_1}^{y}u^{\frac{k}{2}-\frac{k}{4}-\frac{5}{4}}\,du.
\]
The exponent equals $k/4-5/4$. Since $k\geq 2$, the integral is controlled by a constant multiple of $y^{k/4-1/4}$. This yields \eqref{I1bd}.

For $I_2$, use \eqref{g} to obtain
\[
G(\lambda_{G(y)})\leq 2^k\sigma^k\lambda_{G(y)}^{k/2}.
\]
Then
\[
I_2\leq 2^k\sigma^k\lambda_{G(y)}^{\frac{k}{2}-\frac{k}{4}-\frac{1}{4}}.
\]
The exponent equals $k/4-1/4$. Using \eqref{stau} gives \eqref{I2bd}.

For $I_3$, we use \eqref{g} again:
\[
I_3\leq 2^k\sigma^k\int_{y}^{\infty}u^{\frac{k}{2}-\frac{k}{4}-\frac{\rho}{2}-\frac{5}{4}}\,du.
\]
The exponent equals $k/4-\rho/2-5/4$. This is less than $-1$ for $\rho\geq 1$ and $k\geq 2$, so the integral is controlled by a constant multiple of $y^{k/4-\rho/2-1/4}$. This yields \eqref{I3bd}.
\end{proof}

We now package the Abel contributions in a single quantity $I_4$.

\begin{lemma}\label{lem:I4}
Assume Condition \ref{cond:balanced}. Define
\begin{multline}\label{I4def}
I_4=\pi^{\rho}z^{\rho}f\left(\frac{k}{2}+\rho\right)\left[\left(\frac{k}{4}+\frac{1}{4}\right)I_1+I_2\right]\\
+2^{\rho}(x+\rho z)^{\frac{\rho}{2}}f\left(\frac{k}{2}+\rho\right)\left(\frac{k}{4}+\frac{\rho}{2}+\frac{1}{4}\right)I_3.
\end{multline}
Then there exists a constant $c_9$ such that for $x\geq 1$,
\begin{equation}\label{absolute}
\frac{\ln I_4}{k}\leq \ln\sigma+\ln c_9+o(1)+\frac{\rho}{(k+1)k}\ln x.
\end{equation}
Here $o(1)$ depends only on $k$ and satisfies $o(1)\to 0$ as $k\to\infty$.
\end{lemma}

\begin{proof}
We estimate each term in \eqref{I4def} using Lemma \ref{lem:integrals} and Lemma \ref{lem:yest}. The factor $f\left(\frac{k}{2}+\rho\right)$ is at most polynomial in $k$ by Olenko's estimate, so its contribution is $o(1)$ after dividing by $k$. The factor $z^{\rho}$ equals $x^{\rho/(k+1)}$. This contributes the term $\frac{\rho}{(k+1)k}\ln x$. The factor $2^k\sigma^k$ produces the term $\ln\sigma$ plus an absolute constant after dividing by $k$. Collecting these contributions yields \eqref{absolute}.
\end{proof}

We now return to \eqref{38}. From \eqref{new-comment} and \eqref{nnew-comment} with the choices $w=\frac{k}{4}+\frac{1}{4}$ and $w=\frac{k}{4}+\frac{\rho}{2}+\frac{1}{4}$, and using the definition of $I_4$, we obtain
\begin{equation}\label{seriesbound}
\left|\sum_{n=1}^{\infty}\frac{1}{\lambda_n^{\frac{k}{4}+\frac{\rho}{2}}}\Delta\left[x^{\frac{k}{4}+\frac{\rho}{2}}J_{\frac{k}{2}+\rho}\left(2\pi\sqrt{\lambda_n x}\right)\right]\right|
\leq \frac{(x+\rho z)^{\frac{k}{4}-\frac{1}{4}}}{\sqrt{2\pi}}I_4.
\end{equation}

We also use two identities recorded by Landau. On page 475 of \cite{EL},
\begin{equation}\label{deltaxrho}
\Delta(x^{\rho})=\rho !z^{\rho}.
\end{equation}
Moreover, there exists $\xi$ with $x\leq \xi\leq x+\rho z$ such that
\begin{equation}\label{deltaxk}
\Delta\left(x^{\frac{k}{2}+\rho}\right)=z^{\rho}\frac{\Gamma\left(\frac{k}{2}+\rho+1\right)}{\Gamma\left(\frac{k}{2}+1\right)}\xi^{\frac{k}{2}}.
\end{equation}

Plugging \eqref{seriesbound}, \eqref{deltaxrho}, and \eqref{deltaxk} into \eqref{38} and using $\gamma=1$ and $\eta=-1/\rho !$ gives
\begin{multline}\label{DeltaBrho}
\Delta[B_{\rho}(x)]
=\frac{\pi^{k/2}}{\sqrt{D}\Gamma\left(\frac{k}{2}+1\right)}z^{\rho}\xi^{k/2}-z^{\rho} 
+ D^{-1/2}\pi^{-\rho}\frac{(x+\rho z)^{\frac{k}{4}-\frac{1}{4}}}{\sqrt{2\pi}}I_4.
\end{multline}

We compare $\xi^{k/2}$ to $x^{k/2}$. Since $x\leq \xi\leq x+\rho z$, we have
\[
|\xi^{k/2}-x^{k/2}|\leq (x+\rho z)^{k/2}-x^{k/2}.
\]
Using the mean value theorem,
\[
(x+\rho z)^{k/2}-x^{k/2}\leq \rho z\frac{k}{2}(x+\rho z)^{k/2-1}.
\]
Therefore
\begin{multline}\label{bfirst}
\left|\frac{\Delta[B_{\rho}(x)]}{z^{\rho}}-\frac{\pi^{k/2}}{\sqrt{D}\Gamma\left(\frac{k}{2}+1\right)}x^{k/2}\right| \\
\leq \frac{\pi^{k/2}}{\sqrt{D}\Gamma\left(\frac{k}{2}+1\right)}\rho z\frac{k}{2}(x+\rho z)^{k/2-1}
+1+D^{-1/2}\pi^{-\rho}\frac{(x+\rho z)^{\frac{k}{4}-\frac{1}{4}}}{\sqrt{2\pi}}z^{-\rho}I_4.
\end{multline}
Landau notes that $B$ is non-decreasing and that $\Delta$ is a nonnegative averaging operator on non-decreasing functions. In particular,
\[
z^{\rho}B(x)\leq \Delta[B_{\rho}(x)].
\]

We combine this with \eqref{bfirst} to obtain an upper bound on $B(x),$
\begin{multline}\label{bfirstandhalf}
B(x)
\leq 
\frac{\pi^{k/2}}{\sqrt{D}\Gamma\left(\frac{k}{2}+1\right)}x^{k/2}
+
 \frac{\pi^{k/2}}{\sqrt{D}\Gamma\left(\frac{k}{2}+1\right)}\rho z\frac{k}{2}(x+\rho z)^{k/2-1}\\
+1+D^{-1/2}\pi^{-\rho}\frac{(x+\rho z)^{\frac{k}{4}-\frac{1}{4}}}{\sqrt{2\pi}}z^{-\rho}I_4.
\end{multline}
 In order to further modify \eqref{bfirstandhalf}, we record the following arguments.
\begin{itemize}
\item
In our diagonal case, we have $\sqrt{D}=\prod_{i=1}^k \varepsilon_i = (\varepsilon_{\mathrm{geom}})^k.$
\item
Recalling \eqref{relation}, we set $x=R^2.$
\item
As we set $x=R^2,$ 
there exists an absolute constant $c_5$ such that 
\begin{equation}\label{neater-estimate}
    x+\rho z\leq c_5 x.
\end{equation}
This is due to \eqref{rho}, \eqref{z}, and the assumption of theorem \ref{thm:main} 
that we work in the regime \eqref{regime}.
\end{itemize}
Using these arguments, we modify the upper bound in  \eqref{bfirstandhalf} as follows:
\begin{multline}\label{Bupper}
B\left(R^2\right)\leq \frac{\pi^{k/2}}{\left(\varepsilon_{geom}\right)^k\Gamma\left(\frac{k}{2}+1\right)}R^k 
+\frac{\pi^{k/2}}{
\left(\varepsilon_{geom}\right)^k
\Gamma\left(\frac{k}{2}+1\right)}\rho z\frac{k}{2}c_5^{k/2-1}\frac{R^k}{x}
+1\\
+\frac{c_5^{\frac{k}{4}-\frac{1}{4}}x^{\frac{k}{4}-\frac{1}{4}}}{\sqrt{2\pi}\left(\varepsilon_{geom}\right)^k\pi^{\rho}z^{\rho}}I_4
=J_1+J_2+1+J_3.
\end{multline}
We now estimate all three terms: $J_1,J_2,J_3.$
\begin{multline*}
\frac{1}{k}\ln\left(J_1\right)=
\frac{1}{k}\ln\left(\frac{\pi^{k/2}R^k}{ \left(\varepsilon_{geom}\right)^k\Gamma\left(\frac{k}{2}+1\right)}\right)=
\ln R-\ln\varepsilon_{geom}-\frac{1}{2}\ln k+O(1)
\end{multline*}  
Also
\begin{multline*}
\frac{1}{k}\ln\left(J_2\right)=\frac{1}{k}\ln\left(J_1\right)+
\frac{1}{k}\ln\left(
\rho z \frac{k}{2} c_5^{\frac{k}{2}-1}x^{-1}
\right)=
\frac{1}{k}\ln\left(J_1\right)+
\frac{1}{k}\ln\left(
\frac{z}{x}
\right)+O(1)\\=\frac{1}{k}\ln\left(J_1\right)+O(1)
\end{multline*}  
 Finally, we use \eqref{absolute} and the remark after \eqref{sigma}, to obtain
\begin{multline*}
\frac{1}{k}\ln\left(J_3\right)=\frac{1}{k}\ln\left(
\frac{c_5^{\frac{k}{4}-\frac{1}{4}}x^{\frac{k}{4}-\frac{1}{4}}}{\sqrt{2\pi}\left(\varepsilon_{geom}\right)^k\pi^{\rho}z^{\rho}}I_4
\right)
\\
=
\frac{(k-1)\ln x}{4k}-\frac{\rho\ln z}{k}-\ln \varepsilon_{geom}
+
\left(\ln\sigma+\ln c_9+o(1)+\frac{\rho\ln x}{(k+1)k}\right)+O(1)
\\
=\frac{(k-1)\ln x}{4k}-\frac{\rho\ln x}{k(k+1)}-\ln \varepsilon_{geom}
+\frac{\rho\ln x}{(k+1)k}+O(1)
\\
=\frac{k-1}{2k}
\ln R-\ln \varepsilon_{geom}+O(1)
\end{multline*} 
Due to \eqref{regime} the estmates obtained for $J_1$ and $J_2$ dominate that of $J_3.$ As a result,
we obtain the asymptotic behavior
\[
\frac{1}{k}\ln B(R^2)\leq \ln R-\ln\varepsilon_{\mathrm{geom}}-\frac{1}{2}\ln k
+O(1).
\]
This is the estimate claimed in Theorem \ref{thm:main}.

\begin{remark}
The bounded remaining term $O(1)$ can be made explicit using \eqref{absolute} if one wishes to track constants more carefully.
\end{remark}

\begin{remark}
If we consider the regime 
\begin{equation*}
    c_{10}R^{1+\frac{1}{k}}\leq k\leq c_{11}R^2
\end{equation*}
instead of \eqref{regime}, then
the same estimates for $J_1,J_2,J_3$ can be derived. Yet, this time, the estimate for $J_3$ will dominate the other two. As a result, we will obtain the asymptotic behavior
\[
\frac{1}{k}\ln B(R^2)\leq 
\frac{k-1}{2k}\ln R-\ln \varepsilon_{geom}+O(1).
\]
\end{remark}

\begin{remark}
    One might suspect that the bound on $J_3$ could be improved outside of regime \eqref{regime} with a more careful entropy estimate in Section \ref{page-465}, e.g. by estimating the lattice point count with the volume of a $\frac {\sqrt k} 2$-thickened ellipsoid. Alas, the resulting improvement is so small that it gets absorbed into the $O(1)$ error term, and so we elect to keep the simpler bound in Section \ref{page-465}.
\end{remark}

\section*{Acknowledgements}
The authors thank J. Yedoyan for translating E. Landau's article from German.

\end{document}